\documentclass[format=acmsmall, review=false]{acmart}
\usepackage{acm-ec-25}
\usepackage{booktabs} % For formal tables
\usepackage[utf8]{inputenc}
\usepackage[ruled,linesnumbered]{algorithm2e} % For algorithms
\SetKwProg{Init}{Initialization}{}{}
\SetKwFor{If}{if}{then}{} % Custom if-then-end block
\SetKwFor{Else}{else}{}{end if} % Custom else block without end

\SetAlFnt{\small}
\SetAlCapFnt{\small}
\SetAlCapNameFnt{\small}
\SetAlCapHSkip{0pt}
\IncMargin{-\parindent}

\newtheorem{lemma}{Lemma}

\newtheorem{theorem}{Theorem}
\newtheorem{definition}{Definition}

\newtheorem{remark}{Remark}
\usepackage{mathtools}
\DeclarePairedDelimiter\ceil{\lceil}{\rceil}
\DeclarePairedDelimiter\floor{\lfloor}{\rfloor}
\newtheorem{assumption}{Assumption}
\usepackage{subcaption}
\usepackage{natbib}
\setcitestyle{authoryear}

\title[Incentive Analysis for Agent Participation in Federated Learning]{Incentive Analysis for Agent Participation in Federated Learning}

\author{Lihui Yi}
\affiliation{%
  \institution{Northwestern University}
  \country{USA}} 
  \email{lihuiyi2027@u.northwestern.edu}
 \author{Xiaochun Niu}
\affiliation{%
  \institution{Northwestern University}
  \country{USA}} 
  \email{xiaochunniu2024@u.northwestern.edu}
 \author{Ermin Wei}
\affiliation{%
  \institution{Northwestern University}
  \country{USA}} 
  \email{ermin.wei@northwestern.edu}

\begin{abstract}
Federated learning offers a decentralized approach to machine learning, where multiple agents collaboratively train a model while preserving data privacy. In this paper, we investigate the decision-making and equilibrium behavior in federated learning systems, where agents choose between participating in global training or conducting independent local training. The problem is first modeled as a stage game and then extended to a repeated game to analyze the long-term dynamics of agent participation. For the stage game, we characterize the participation patterns and identify Nash equilibrium, revealing how data heterogeneity influences the equilibrium behavior---specifically, agents with similar data qualities will participate in FL as a group. We also derive the optimal social welfare and show that it coincides with Nash equilibrium under mild assumptions. In the repeated game, we propose a privacy-preserving, computationally efficient myopic strategy. This strategy enables agents to make practical decisions under bounded rationality and converges to {a neighborhood of Nash equilibrium} of the stage game in finite time. 
By combining theoretical insights with practical strategy design, this work provides a realistic and effective framework for guiding and analyzing agent behaviors in federated learning systems. 
\end{abstract}

\begin{document}

% Title page for title and abstract only.
\begin{titlepage}

\maketitle
\makeatletter \gdef\@ACM@checkaffil{} \makeatother

% Optionally include a table of contents
\vspace{1cm}
\setcounter{tocdepth}{2} % adjust to 1 if desired
\tableofcontents

\end{titlepage}

\section{Introduction}
Federated learning (FL) has been revolutionary in modern machine learning, addressing growing concerns about data privacy and decentralization. By enabling agents to collaboratively train a global model without directly transferring raw data, FL has found applications across a wide range of industries. For instance, it is used in healthcare to build predictive models from distributed hospital datasets without compromising patients' confidentiality \cite{FL_healthcare, FL_healthcare2, FL_healthcare3, FL_healthcare4}; FL also improves personalized user experience in mobile and IoT devices without exporting sensitive user data to servers \cite{keyboard, emoji}; in vehicular networks, FL enables autonomous vehicles to jointly learn and optimize their controls without sharing their history trajectories \cite{autonomous}.

In addition to advances in federated learning algorithms \cite{mcmahan2017, li2019, collins2022} and generalizations \cite{mohri2019agnostic, lin2020ensemble, fallah2021generalization, zhu2021data, sefidgaran2022rate}, the economic perspective of FL is equally important and needs research attention. While FL aims to harness data distributed across numerous agents, the whole system may not always benefit from the participation of every single agent \cite{MinghongLocal2020}. On the other hand, individual agents may not always benefit from engaging in the FL process either \cite{sheller_2020_federated}. Therefore, mechanism design and incentive analysis play a crucial role in sustainable FL systems.

A substantial body of research has explored incentive mechanisms by modeling the costs of data sharing, communication, and computation in FL, but without considering data quality.
% \ew{mention the costs here models data sharing, communication, computation, rather than data quality. The first sentence in the next paragraph should be hinted in this paragraph as well}
Specifically, %Murhekar et al. 
\cite{ruta2023} examines agents' trade-offs between the cost of data sharing and FL benefit through a game-theoretic framework, demonstrating the existence of Nash equilibrium and propose a budget-balanced mechanism to maximize welfare across agents. %Kang et al. 
\cite{Kang2024} studies a multi-server selection game in FL with the consideration of users' communication costs and server handover costs, to optimize the system's energy consumption. %Zhang et al. 
\cite{Zhang2023enabling} considers agents' computation and communication costs, and model the long-term FL participation by an infinitely repeated game. They propose a cooperative strategy that is a subgame-perfect equilibrium and minimizes the number of free riders. Similarly, %Bi et al. 
\cite{bi2024understanding} models the formation and long-term partnerships in FL by an iterated prisoner's dilemma, where they consider an aggregated linear cost of data sharing. Note that the latter two works both focus on long-term FL participation, however, neither considers the possibility of agents leaving the system.

Though many research works consider the cost of data sharing, communication, and computation, the study of incentives related to data quality is mostly missing. Experimental evidence shows that agents' heterogeneous data distributions can largely influence their preferences for FL participation, as well as the system's preference for retaining the agents. In particular, %Sheller et al. 
\cite{sheller_2020_federated} experiments with the task of distinguishing healthy brain tissue from tissue affected by cancer cells. Their results reveal that some FL agents achieved higher prediction accuracy through local training compared to participating in FL. This discrepancy is attributed to differences in data quality among agents. For some agents, the global FL model is trained on a dataset that is too different from their own, they may rather not participate in FL and train a local model instead. Observing the importance of data quality factor, %Donahue and Kleinberg 
\cite{Donahue_Kleinberg_2021} considers heterogeneous players with different data distributions in a one-shot game to study how they might divide into coalitions. However, they focus on linear regression and mean estimation problems in particular and do not consider long-term FL participation. Therefore, our work is driven by the lack of study on data quality consideration in long-term FL participation with agents allowed to leave the system.

In this paper, we address the fundamental question about agent participation in federated learning, focusing on the impact of data quality in both the short term and the long term. Specifically, we aim to understand the following aspects: a) How differences in data quality influence agents' decisions on opting in FL; b) What strategy maximizes social welfare and how achievable it is in decentralized environments; c) In repeated decision-making, with incomplete and imperfect information revealed, what strategy may guide the system into a stable state? 

To answer these questions, we first model a stage game of FL participation, where agents' payoffs/costs are derived from %bounds on 
the performance of FL algorithms, with explicit modeling of data qualities. Then, we establish the existence of Nash equilibrium and characterize the equilibrium state, showing that agents with similar data qualities will participate in FL as a group. Moreover, we find that the optimal social welfare strategy coincides with the Nash equilibrium under mild assumptions, demonstrating the alignment between decentralized strategic behaviors and system efficiency. Upon this, we extend the stage game into an infinitely repeated game, to capture the long-term dynamics of FL participation. We propose a privacy-preserving myopic strategy, which {drives the system to a neighborhood of Nash equilibrium} of the stage game in finite time. 
The main contribution of this work is summarized as follows,
\begin{itemize}
    \item \textbf{Game-Theoretic Modeling}: we model a stage game of FL participation, with payoffs explicitly connected to agents' data qualities. {Later on, we extend it into a repeated game to study the long-term interactions.}
    \item \textbf{Nash Equilibrium and Social Welfare Maximization}: {we characterize the Nash equilibrium and the social welfare maximization strategy, and find the alignment between the two strategies.}
    \item \textbf{Simple and Efficient Learning Dynamics}: we propose a simple myopic strategy that {guides the system close to an equilibrium in finite time} while preserving privacy and requiring minimal local computational resources {and no central coordinator}. 
\end{itemize}

\section{Stage Game}

\subsection{Problem Formulation}

In this section, we investigate a finite stage game $\mathcal{G}$ with $m$ agents possibly participating in a federated learning system, where each agent $i \in \{1,2,...,m\}$ has $n > 0$ data points available following a distribution $\mathcal{D}_i$ with mean $\mu_i$. The order of agents is arranged such that $\mu_1 \le \mu_2 \le \cdots \le \mu_m$. For simplicity, we assume the means are evenly spread out, i.e., $\Delta \coloneq \mu_m - \mu_{m-1} = \cdots = \mu_2 - \mu_1$ and $\Delta \ge 0$. {The parameter $\Delta$ is referred to as the \textit{data separation} or \textit{mean separation}.}

Agents aim to learn a good prediction model either collaboratively through FL or individually through local machine learning. That is to say, each player $i$ has a {binary} action set $S_i = \{0, 1\}$. They may choose to perform local training, {i.e., opt out,} with strategy $s_i=0$; or to participate in federated learning leveraging information about all available data within the federated system  {, i.e., opt in, } with $s_i=1$. We denote the set of players who {opt in} by the \textit{participant set} $\Omega \coloneq \big\{i \in \{1,2,...,m\}:\, s_i = 1\big\}$. The player in the participant set is called an \textit{FL agent}, and the player not in the participant set is called a \textit{non-FL agent}.

Let $S = S_1 \times \cdots \times S_m$. For $s = (s_1,\dots,s_m) \in S$, denote $s_{-i}$ as the profile of player strategies other than player $i$, i.e., 
$$s_{-i} = (s_1,\dots,s_{i-1},s_{i+1},\dots,s_m).$$
With this notation, we will use $s$ and $(s_i, s_{-i})$ interchangeably to represent the strategy profile of all players. {We also call it the \textit{state} of the system without ambiguity.}

In our model, each agent $i$ aims to minimize their local expected loss, which can be approximated by a cost function $c_i: S \to \mathbb{R}_{\ge 0}$. Specifically, {we use } the performance bounds of FL {as cost function}, with the explicit connection to data qualities. If player $i$ participates in federated learning, i.e., $s_i=1$, then we define the cost function similar as \cite{Meng},
\begin{align}\label{eq: cost join}
    c_i(s_i=1, s_{-i}) \coloneq \frac{a}{|\Omega| \cdot n} + \left| \mu_i - \frac{\sum_{j \in \Omega}\mu_j}{|\Omega|} \right|.
\end{align}
Here, $a \in \mathbb{R}_{>0}$ is a constant determined by {the specific FL problem}, the first term accounts for the {total} number of data points in the FL system, and the second term captures the difference between the local distribution and the averaged distribution across all participants. {When there is a large number of data available in FL system, the prediction accuracy increases; therefore the cost decreases. For agents with data similar to those in the system, the federated model generalizes well to their local data distributions, leading to lower loss.} On the other hand, if player $i$ performs local training instead of joining the FL system, i.e., $s_i = 0$, then we define the cost function as
\begin{align}\label{eq: cost not join}
    c_i(s_i=0, s_{-i}) \coloneq \frac{a}{n}.
\end{align}
Note that depending on player $i$'s strategy $s_i$, the cost function may or may not be independent of other players' strategies $s_{-i}$.

\subsection{Nash Equilibrium}
To study the strategic interactions in a complex system with self-interested agents, one key solution concept is \textit{Nash equilibrium}, which describes a stable state where no agent benefits from unilaterally changing their strategy. Formally, in our context, we have the following definition.
\begin{definition}
    A strategy profile $s^* \in S$ is called a pure Nash equilibrium\footnote{We will focus on pure Nash equilibrium in this paper, and henceforth refer to a pure Nash equilibrium simply as an equilibrium.} if, for every player $i \in \{1,2,...,m\}$, and every $s'_i$, we have $c_i(s_i^*, s_{-i}^*) \le c_i(s'_i, s_{-i}^*)$.\footnote{Note that the definition is slightly different from the convention. This is because the players aim to minimize their cost functions in our setup, rather than maximizing payoffs.}
\end{definition}

We observe the following characteristics of any equilibrium, if any exists.

\begin{lemma}\label{lemma: grouping feature of NE}
Assume $s^* =(s_1^*, \cdots, s_m^*)$ is an equilibrium of the stage game $\mathcal{G}$. Fix $i\in\{1,2,...,m\}$, if there exist $p,q\in \{1,2,...,m\}$ such that $p<i<q$  and $s_p^* = s_q^* = 1$, then $s_i^* = 1$.
\end{lemma}
\begin{proof}
For the equilibrium strategy $s^*$, we denote the corresponding participant set by $\Omega^*$. Then, $s_p^* = 1$ implies, 
$$c_p(s_p^*=1, s_{-p}^*)  = \frac{a}{|\Omega^*|\cdot n} + \left|\mu_p - \frac{\sum_{j\in\Omega^*}\mu_j}{|\Omega^*|}\right| 
 \le c_p(s_p=0, s_{-p}^*) = \frac{a}{n}.$$
Similarly, $s_q^*=1$ implies that
$$c_q(s_q^*=1, s_{-q}^*) = \frac{a}{|\Omega^*|\cdot n} + \left|\mu_q - \frac{\sum_{j\in\Omega^*}\mu_j}{|\Omega^*|}\right| 
 \le c_q(s_q=0, s_{-q}^*) = \frac{a}{n}.$$
Since $p<i<q$, we have $\mu_p \le \mu_i \le \mu_q$ and thus 
$$\left|\mu_i - \frac{\sum_{j\in\Omega^*}\mu_j}{|\Omega^*|}\right| \le \max\left\{\left|\mu_p - \frac{\sum_{j\in\Omega^*}\mu_j}{|\Omega^*|}\right|, \left|\mu_q - \frac{\sum_{j\in\Omega^*}\mu_j}{|\Omega^*|}\right|\right\}.$$ 
We then study the two cases $s_i=1$ and $s_i=0$ separately.
\paragraph{Case (I)} Suppose $s_i=1$, then we have $p,i,q \in \Omega^*$ and thus
\begin{align*}
    c_i(s_i=1, s_{-i}^*) & = \frac{a}{|\Omega^*|\cdot n} + \left|\mu_i - \frac{\sum_{j\in\Omega^*}\mu_j}{|\Omega^*|}\right| \\
    & \le \max\left\{c_p(s_p^*=1, s_{-p}^*), c_q(s_q^*=1, s_{-q}^*) \right\} \\
    & \le \max\left\{c_p(s_p=0, s_{-p}^*), c_q(s_q=0, s_{-q}^*) \right\} \\
    & = \frac{a}{n} = c_i(s_i=0, s_{-i}^*),
\end{align*}
which proves that agent $i$ has no incentive to deviate from $s_i = 1$. 

\paragraph{Case (II)} Suppose $s_i=0$, then we have $p,q\in\Omega^*$ but $i\notin\Omega^*$. Thus, 
\begin{align*}
    c_i(s_i=1, s_{-i}^*) & = \frac{a}{(|\Omega^*| +1)n} + \left|\mu_i - \frac{\mu_i+\sum_{j\in\Omega^*}\mu_j}{|\Omega^*| +1}\right| \\
    & = \frac{|\Omega^*|}{|\Omega^*| +1} \cdot \left(\frac{a}{|\Omega^*|\cdot n} + \left|\mu_i - \frac{\sum_{j\in\Omega^*}\mu_j}{|\Omega^*|}\right| \right)\\
    &< \frac{a}{|\Omega^*|\cdot n} + \left|\mu_i - \frac{\sum_{j\in\Omega^*}\mu_j}{|\Omega^*|}\right|\\
    &\le \max\left\{c_p(s_p=0, s_{-p}^*), c_q(s_q=0, s_{-q}^*) \right\} \\
    & = \frac{a}{n} = c_i(s_i=0, s_{-i}^*).
\end{align*}
This implies that agent $i$ benefits from unilaterally deviation. Hence, if $s^*$ is an equilibrium point, then $s_i^*=1$.
\end{proof}

Essentially, the lemma indicates that, in any equilibrium state (if exists), agents with similar data qualities will make the same decision on joining or not joining FL. Formally, we capture this ``grouping'' feature by the following notions.
\begin{definition}
    We call $\mathcal{P} \subseteq \{1,2,...,m\}$ a $k$-consecutive set if $|\mathcal{P}| = k$ and, for any $p,q \in \mathcal{P}$, there holds $i \in \mathcal{P}$ for every $i \in \{1,2,...,m\}$ such that $p<i<q$. 
\end{definition}
\begin{definition}
    We say a strategy profile $s \in S$ forms a $k$-consecutive participation if there exists a $k$-consecutive set $\mathcal{P} \subseteq \{1,2,...,m\}$ such that $s_i = 1$ for any $i \in \mathcal{P}$ and $s_j = 0$ for any $j \notin \mathcal{P}$. In other words, the participant set $\Omega$ is a $k$-consecutive set. 
\end{definition}
We then show the existence of the Nash equilibrium and provide its explicit characterization. 
% \ew{All looks correct so far.}
\begin{theorem}\label{thm: NE}
    For the stage game $\mathcal{G}$, Nash equilibria exist and have the following two types:
    % \begin{enumerate}
    %     \item type 1 (0-consecutive participation): $s^* = 0$;
    %     \item type 2 ($k^*$-consecutive participation): $s^*$ forms a $k^*$-consecutive participation.
    % \end{enumerate}
    \paragraph{type 1 (0-consecutive participation)} $s^* = 0$;
    \paragraph{type 2 ($k^*$-consecutive participation)} $s^*$ forms a $k^*$-consecutive participation, where
    \begin{align}\label{eq: k* explicit form}
        k^* = \begin{dcases}
            m, \quad\text{if }\; 0 \le \Delta < \frac{2a}{mn};\\
            \ceil*{\frac{2a}{n\Delta}-1} \text{ or } \floor*{\frac{2a}{n\Delta}}, \quad\text{if }\; \frac{2a}{mn} \le \Delta \le \frac{a}{n};\\
            1, \quad\text{if }\; \Delta > \frac{a}{n}.
        \end{dcases}
    \end{align}
    {Moreover, any strategy $s$ that forms a $k^*$-consecutive participation with $k^*$ satisfying Eq. \eqref{eq: k* explicit form} is a type 2 Nash equilibrium.}
\end{theorem}
\begin{proof}
    We first show the existence of type 1 equilibrium. Consider $s^* = (s_1^*, \dots, s_m^*) = (0, \dots, 0)$ in any stage game $\mathcal{G}$, if some agent $i$ deviates to $s_i = 1$, their cost becomes $c_i(s_i=1, s_{-i}^*) = a/n$, which is the same as $c_i(s_i=0, s_{-i}^*)$. Hence, $s^*=0$ is an equilibrium by definition. 

    By Lemma \ref{lemma: grouping feature of NE}, it is straightforward to see that any equilibrium $s \neq 0$, if it exists, forms a $k$-consecutive participation for some $k\in\{1,2,...,m\}$. Next, we validate that $s$ is indeed a Nash equilibrium if and only if $k$ satisfies Eq. \eqref{eq: k* explicit form}. Since the cost functions are inherently different for an FL agent and a non-FL agent, we consider the behaviors of FL agents and non-FL agents separately.
    \paragraph{Case (I)} 
    For any $i\in\Omega$, i.e., $s_i=1$, since the strategy profile $s$ forms a $k$-consecutive participation, we have $c_i(s_i=1,s_{-i}) = \frac{a}{k\cdot n} + \big|\mu_i - \frac{\sum_{j\in\Omega}\mu_j}{k}\big|$. We observe that the FL agent who bears the largest cost is the one with the largest or smallest index in the participant set. That is to say, 
    $$\max_{i\in\Omega} c_i(s_i=1,s_{-i}) = \frac{a}{kn} + \frac{(k-1)\Delta}{2}.$$
    Thus, no FL agent has incentive to deviate from $s_i=1$ to $s'_i=0$ if and only if 
    $$\max_{i\in\Omega} c_i(s_i=1,s_{-i}) \le a/n = c_i(s'_i=0,s_{-i}),$$
    which always holds if $k=1$, and is equivalent to $\Delta \le \frac{2a}{nk}$ otherwise.

    \paragraph{Case (II)} 
    For any $i \notin \Omega$, i.e., $s_i=0$ and $k \neq m$, their cost after deviating to $s'_i=1$ is 
    $$c_i(s'_i=1,s_{-i}) = \frac{a}{(k+1)n} + \left|\mu_i - \frac{\mu_i+\sum_{j\in\Omega}\mu_j}{k+1}\right|.$$
    We observe that the non-FL agent who bears the lowest cost if joining the FL is the one(s) whose index is the closest to that of an arbitrary FL agent. In other words, 
    $$\min_{i \notin \Omega} c_i(s'_i=1,s_{-i}) = \frac{a}{(k+1)n} + \frac{k\Delta}{2}.$$
    Thus, no non-FL agent has incentive to deviate from $s_i=0$ to $s'_i=1$ if and only if
    $$\min_{i \notin \Omega} c_i(s'_i=1,s_{-i}) \ge a/n = c_i(s_i=0,s_{-i}),$$
    which is equivalent to $\Delta \ge \frac{2a}{n(k+1)}$. 

% \ew{where do the first two inequalities come from?} 

{The strategy profile $s$ is an equilibrium if and only if both conditions in the two cases established above are satisfied, i.e.,}
    \begin{align*}{
        \begin{dcases}
            \Delta \ge \frac{2a}{n(k+1)} = \frac{a}{n}, \quad\text{for}\;\; k=1;\\
            \Delta \le \frac{2a}{nk} = \frac{2a}{nm}, \quad\text{for}\;\; k=m;\\
            \frac{2a}{n(k+1)} \le \Delta \le \frac{2a}{nk}, \quad\text{otherwise}.
        \end{dcases}}
    \end{align*}
    Clearly, the conditions above are well-defined. Since $k\in\mathbb{N}_{>0}$ and $\Delta \ge 0$, we may write the explicit form of $k$ as follows,
    \begin{align*}
        k = \begin{dcases}
            1, \quad\text{if}\;\; \Delta > a/n;\\
            m, \quad\text{if}\;\; 0 \le \Delta < 2a/(nm);\\
            \ceil*{\frac{2a}{n\Delta}-1} \text{ or } \floor*{\frac{2a}{n\Delta}}, \quad\text{otherwise},
        \end{dcases}
    \end{align*}
    which is precisely Eq. \eqref{eq: k* explicit form}. 
\end{proof}

The two types of equilibrium for the stage game are rather intuitive. For type 1 equilibrium, if every player opts out of the FL system and one unilaterally deviates, the one who joins cannot gain any benefit since no extra data samples are shared with the server. Meanwhile, those who stay out will not be affected since they all perform local training. Thus, no one participating in the FL is always an equilibrium strategy. For type 2 equilibrium, the intuition is that agents with similar data qualities can collaboratively reduce their costs by joining the FL, and those with very distinct distributions cannot benefit from the global model as it is trained on a {dataset that is too different from their own}. {For the same reason, the more spread out the agents' distribution means, the fewer FL participants there will be in a type 2 equilibrium. Moreover, more data samples that every agent owns also result in fewer FL participants, as many agents become more willing to train a better local model with the increased sample points.}
From now on, we will use $k^*$ to denote the number of participants in a type 2 equilibrium.

It is important to note that the stage game $\mathcal{G}$ may have multiple type 2 equilibria, as there is more than one $k^*$-consecutive set if $k^* < m$. Moreover, there might be $k_1 \neq k_2$ such that $s^*$ is an equilibrium with either $k_1$ or $k_2$ participants. {For simplicity, we restrict our attention to cases where type 2 equilibrium has a unique and odd number of FL participants. That is to say, if $\frac{2a}{nm} \le \Delta \le \frac{a}{n}$, then $k ^* = \ceil*{\frac{2a}{n\Delta}-1} = \floor*{\frac{2a}{n\Delta}}$ and is odd.} %an odd number of participants, i.e., $k^*$ is odd. \xn{I get confused here. Do we mean that whenever we have two possible $k_1 \neq k_2$, we choose $k^*$ to be the odd number in $\{k_1, k_2\}$?} \ly{I meant to say, we consider those parameters such that there is a unique $k^*$ and that $k^*$ is an odd number.} 
Formally, we have the following assumption,
\begin{assumption}\label{assumption: odd k}
    {We consider cases where the type 2 equilibrium of the stage game $\mathcal{G}$ has a unique and odd number of FL participants $k^*$.}
\end{assumption}
The assumption is achievable in a wide range of parameters $\Delta$. Fig. \ref{fig: k-Delta} demonstrates $k^*$ with respect to $\Delta$, where odd $k^*$ is shown in red and even $k^*$ is shown in blue. Under Assumption \ref{assumption: odd k}, the total number of equilibria in the stage game is $m-k^*+2$, including one type 1 equilibrium and $m-k^*+1$ type 2 equilibria. 
% \ew{THere are two values for $k^*$ one floor, one ceiling, doesn't that guarantee one of them is odd? Did you plot both of them for Figure 1?}

\begin{figure}[h]
  \centering
  \includegraphics[width=0.38\linewidth]{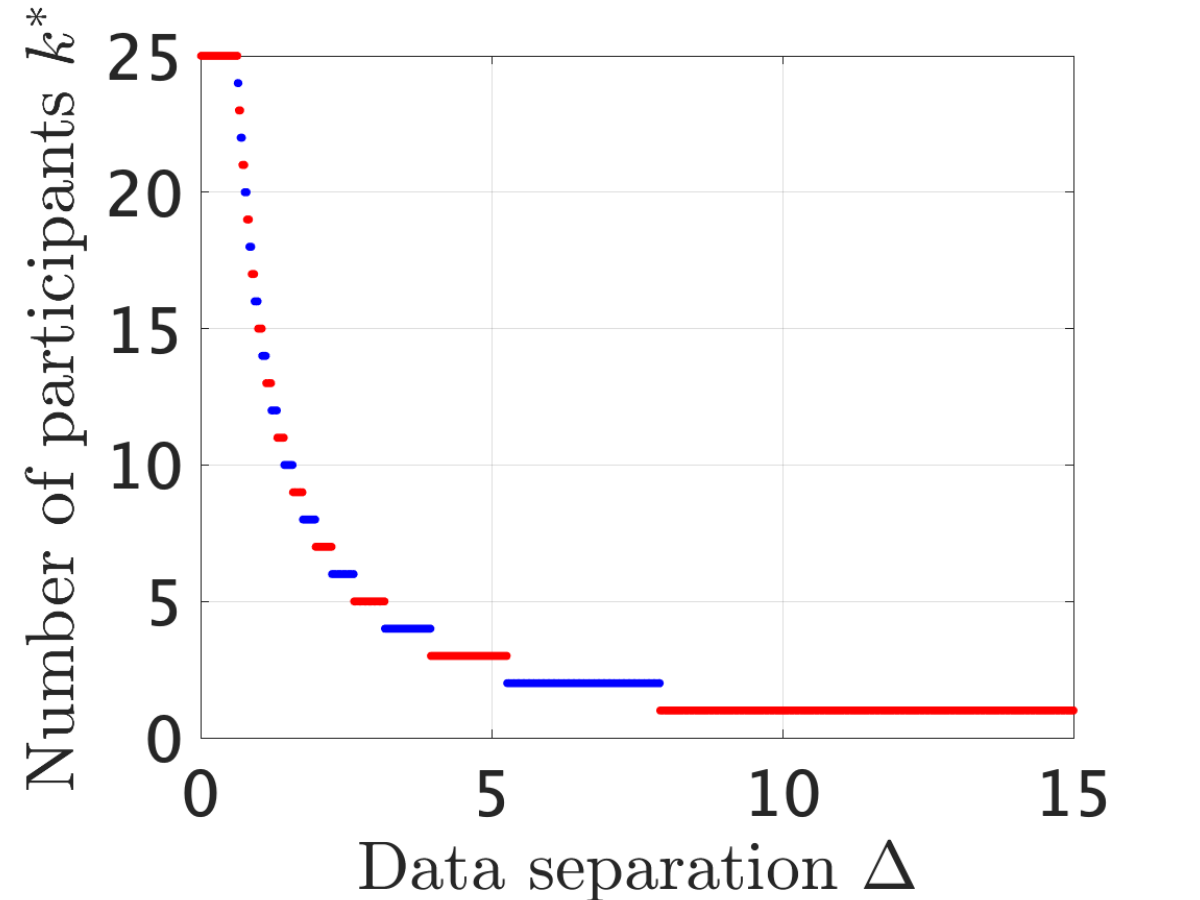}
  \caption{The relation between the number of participants in a type 2 equilibrium $k^*$ and the data separation $\Delta$. The following parameters are used for the plot: $m=25$, $n=100$, $a=790$.}
  \label{fig: k-Delta}
\end{figure}

\subsection{Social Welfare Maximization}\label{subsec: social welfare}

In a decentralized system where agents engage in strategic decision-making, the Nash equilibrium reveals a stable state, as no individual agent has an incentive to unilaterally deviate. However, an equilibrium does not necessarily guarantee that the system is efficient from a social welfare perspective. Therefore, a natural question to ask is how close an equilibrium state is to achieving optimal social welfare, which reflects the collective benefit/cost for all agents in the system. To address this, we investigate the strategy that maximizes social welfare and analyze its relationship to the type 2 equilibrium identified in the stage game $\mathcal{G}$. First of all, we define \textit{social welfare} as $W(s) = \sum_{i=1}^{m}-c_i(s_i, s_{-i})$, also known as the utilitarian welfare function in economic theory \cite{welfare_function}. The following lemma provides a concrete expression of $W(s)$.

\begin{lemma}\label{lemma: welfare formula}
    For any strategy profile $s$ with the corresponding participant set $\Omega$, the social welfare can be expressed as
    \begin{align*}
        W(s) = \big( |\Omega|-m-1 \big)\frac{a}{n} - f(\Omega),
    \end{align*}
    where $f(\Omega) \coloneq \sum_{i\in\Omega}\big| \mu_{i}-\frac{\sum_{j\in\Omega}\mu_j}{|\Omega|} \big|$.
\end{lemma}
\begin{proof}
    For $\forall i \in \Omega$, the cost function is $c_i(s_i=1, s_{-i}) = \frac{a}{|\Omega| \cdot n} + \big| \mu_i - \frac{\sum_{j \in \Omega}\mu_j}{|\Omega|} \big|$. By summing up the costs among $\Omega$, we obtain 
    $$\sum_{i\in\Omega} c_i(s_i,s_{-i})= \frac{a}{n} + \sum_{i\in\Omega}\Big| \mu_{i}-\frac{\sum_{j\in\Omega}\mu_j}{|\Omega|} \Big|.$$
    For $\forall i \notin \Omega$, the cost function is $c_i(s_i=0, s_{-i}) = \frac{a}{n}$. By summing up all $c_i$, we obtain
    $$\sum_{i\notin\Omega} c_i(s_i,s_{-i}) = \big( m-|\Omega| \big)\frac{a}{n}.$$
    Hence, by definition of the social welfare, we have
    $$W(s) = -\sum_{i\in\Omega} c_i(s_i,s_{-i}) - \sum_{i\notin\Omega} c_i(s_i,s_{-i}) = \big( |\Omega|-m-1 \big)\frac{a}{n} - \sum_{i\in\Omega}\Big| \mu_{i}-\frac{\sum_{j\in\Omega}\mu_j}{|\Omega|} \Big|. $$ 
\end{proof}

Next, we identify the optimal strategy profile that maximizes social welfare and demonstrate its equivalence to a type 2 equilibrium of the stage game $\mathcal{G}$. 

\begin{theorem}\label{thm: social welfare maximization}
    A strategy profile $s$ achieves the maximum social welfare if and only if it is a type 2 equilibrium of the stage game $\mathcal{G}$, {assuming the number of FL participants in the type 2 equilibrium $k^*$ is unique and odd.}
    % \ew{does this matter if it's the floor or the ceiling of type 2? Also are we assuming $k^*$ is odd? It would be good to remind the readers}
\end{theorem}
\begin{proof}
    For strategy profile $s$ with the corresponding participant set $\Omega$, we observe from Lemma \ref{lemma: welfare formula} that the social welfare $W(s)$ is a function of $\Omega$. Specifically, the first term is only related to the number of participants in the FL system, whereas the second term is dependent on the specific elements in $\Omega$. Here, we first fix $|\Omega| = k$ for any
    $k \in \{1,2,...,m\}$ and solve the subproblem: $\min_{\Omega\subseteq \{1,2,...,m\}, |\Omega| = k} f(\Omega)$.

    Denote $\Bar{\mu} = \frac{1}{|\Omega|}\sum_{j\in \Omega}\mu_j$, the subproblem is rewritten as,
    \begin{align*}
        \min_{\Omega\subseteq \{1,2,...,m\}, |\Omega|=k} \quad f(\Omega)= \sum_{i\in\Omega}| \mu_{i}-\Bar{\mu} |
    \end{align*}
    The best possible minimal value of $f(\Omega)$ is when each individual term is minimized. Since $|\Omega| = k$ and $\mu_m - \mu_{m-1} = \cdots = \mu_2 - \mu_1$, it is equivalent to selecting a $k$-consecutive set of agents. Thus, the optimal solution of the preceding subproblem is any $\Omega$ that is a $|\Omega|$-consecutive set, where the optimal objective value is 
    \begin{align*}
        f^*(\Omega) = \begin{dcases}
            \frac{|\Omega|^2-1}{4}\Delta, &\quad\text{if $|\Omega|$ is odd};\\
            \frac{|\Omega|^2}{4}\Delta, &\quad\text{if $|\Omega|$ is even}.
        \end{dcases}
    \end{align*}
    Therefore, we see that the minimum value of $f(\Omega)$ is a function of $|\Omega|$. According to Lemma \ref{lemma: welfare formula}, to maximize $W(s)$, the best possible solution is maximizing the first term $\big( |\Omega|-m-1 \big)\frac{a}{n}$ and meanwhile minimizing the second term $f(\Omega)$. Hence, we can equivalently solve the social welfare maximization problem $\max_s W(s)$ by substituting $f^*(\Omega)$ into $W(s)$, i.e., 
    \begin{align*}
        W(s) = \begin{dcases}
            \big( |\Omega|-m-1 \big)\frac{a}{n} - \frac{|\Omega|^2-1}{4}\Delta, &\quad\text{if $|\Omega|$ is odd};\\
            \big( |\Omega|-m-1 \big)\frac{a}{n} - \frac{|\Omega|^2}{4}\Delta, &\quad\text{if $|\Omega|$ is even}.
        \end{dcases}
    \end{align*}
    Thus, we observe that the original optimization problem becomes $\max_{|\Omega| \in \{1,2,...,m\}} W(s)$. Note that, without the integer constraint of $|\Omega|$, the maximizer for both cases is $\frac{2a}{n\Delta}$. However, with the constraint added, we {maximize over} the following two cases separately.
    \paragraph{Case (I)} For even $|\Omega|$, the best possible maximum $W^*(s)$ is achieved when $\frac{2a}{n\Delta}$ is an even number, i.e., 
    $$W^*(s) \le W(s)\big|_{|\Omega|=\frac{2a}{n\Delta}} = -(m+1)\frac{a}{n}+\frac{a^2}{n^2\Delta},$$
    with equality if and only if $|\Omega|=\frac{2a}{n\Delta}$ and $\frac{2a}{n\Delta}$ is even.
    \paragraph{Case (II)} For odd $|\Omega|$, the worst possible maximum $W^*(s)$ is achieved when $\frac{2a}{n\Delta}$ is an even number as well, i.e., 
    % \ew{what is the worst case max? And why is it enough to maximize the lower bound? }
    $$W^*(s) \ge W(s)\big|_{|\Omega|=\frac{2a}{n\Delta} \pm 1} = -(m+1)\frac{a}{n}+\frac{a^2}{n^2\Delta},$$
    with equality if and only if $|\Omega|=\frac{2a}{n\Delta} \pm 1$ and $\frac{2a}{n\Delta}$ is even. {This shows that the worst possible maximum $W^*(s)$ for odd $|\Omega|$ is still better or equal to the best possible maximum $W^*(s)$ for even $|\Omega|$, indicating an odd $|\Omega|$ is generally preferred to maximize the social welfare.}

    Therefore, the optimal solution $|\Omega|^*$ is to round $\frac{2a}{n\Delta}$ either up or down, whichever is an odd number, except when $\frac{2a}{n\Delta}$ is an even number, then $\frac{2a}{n\Delta}$ and $\frac{2a}{n\Delta} \pm 1$ are all optimal. Since we assume type 2 equilibrium of the stage game $\mathcal{G}$ has an odd number of participants, based on Theorem \ref{thm: NE}, $\floor*{\frac{2a}{n\Delta}}$ must be an odd number. Hence, we have $|\Omega|^*=\floor*{\frac{2a}{n\Delta}}=k^*$, where $k^*$ is the number of participants in a type 2 equilibrium of the stage game $\mathcal{G}$. The proof is then complete. 
\end{proof}

This result highlights that decentralized agent behaviors, as characterized by Nash equilibrium, can align with the system's efficiency, as defined by social welfare maximization. It provides a theoretical foundation for strategy design, revealing the possibility of steering agents toward optimal social welfare outcomes, even in decentralized environments. Thus, in the context of long-term FL participation, achieving optimal social welfare only requires designing strategies that guide agents to an equilibrium.

\begin{remark}
    While the equivalence holds in our model, we acknowledge the concerns about its robustness in more complex federated learning settings. For instance, the assumption of evenly spread out means might influence both equilibrium and optimal social welfare strategy. We leave the question to future investigation. 
\end{remark}

\section{Repeated Game}

{To model} long-term federated learning, {we consider the setup where} agents repeatedly decide whether to participate in FL. After each training and reward period, agents observe the past cost(s) incurred and reconsider their decisions on FL participation. This dynamic decision-making process can be captured by an infinitely repeated game, denoted as $\mathcal{G}^\infty$. In this section, we model this repeated interaction and propose {a} simple dynamic learning rule that guides the system {close to} both equilibrium and the optimal social welfare state of the stage game $\mathcal{G}$. 

\subsection{Problem Formulation}
First, we extend the stage game $\mathcal{G}$ into an infinitely repeated game $\mathcal{G}^\infty$, where agents repeatedly interact over discrete time stages $t \in \{0,1,2,\dots\}$. At each stage $t$, the action set of player $i$ is denoted by $S_i^t$, and the joint action set of all players is $S^t = S_1^t\times\cdots\times S_m^t$. The strategy profile of all players (state of the system) is denoted by $s^t = (s_i^t,s_{-i}^t)$ where $s_i^t$ is the strategy of player $i$, and $s_{-i}^t$ is the strategy of all other players. The participant set at time $t$ is then defined as $\Omega^t \coloneq \big\{i \in \{1,2,...,m\}: \, s_i^t = 1\big\}$. Similar to Eqs. \eqref{eq: cost join} and \eqref{eq: cost not join}, the cost (negative payoff) for players depends on their strategy. If player $i$ opts in FL at time $t$,
$$c_i^t(s_i^t=1, s_{-i}^t) \coloneq \frac{a}{|\Omega^t|n} + \left| \mu_i - \frac{\sum_{j \in \Omega^t}\mu_j}{|\Omega^t|} \right|.$$
If player $i$ opts out of FL and performs local training,
$$c_i^t(s_i^t=0, s_{-i}^t) \coloneq \frac{a}{n}.$$
% The overall cost for player $i$ over the entire game is given by
% $$C_i \coloneq (1-\delta)\sum_{t=1}^\infty \delta^{t-1}c_i^t(s_i^t,s_{-i}^t),$$
% where $\delta \in [0,1)$ is the discount factor. 

\subsection{Subgame-Perfect Equilibrium}
To establish a baseline, we consider the ideal case where players have complete and perfect information. This includes knowledge of the number of players, the distribution means of others, their actions and payoffs, etc. Under these conditions, we show that the outcome of repeated equilibrium plays of the stage game is supported by a subgame-perfect equilibrium \cite{fudenberg_1991_game}.

\begin{theorem}
    The following grim-trigger strategy is subgame perfect: for every player $i$,
    \begin{enumerate}
        \item play a type 2 equilibrium strategy $s_i^*$ of the stage game $\mathcal{G}$ in the first stage and continue to play $s_i^*$ if no one deviates;
        \item if someone deviates, play a type 1 equilibrium strategy from that point onward.
    \end{enumerate}
\end{theorem}

The proof is straightforward: since players always follow a Nash equilibrium strategy at any time, no agent can reduce their cost by unilaterally deviating from the grim-trigger strategy. This ensures subgame perfection based on the one-stage deviation principle. %In summary, the outcome of this subgame-perfect equilibrium is that players play an equilibrium strategy of the stage game in every single stage. 

While multiple subgame-perfect equilibria may exist since the stage game has multiple Nash equilibria, the key takeaway is that %players repeat a type 2 equilibrium strategy of the stage game in every single stage is a subgame-perfect equilibrium. %
the outcome of repeated type 2 equilibrium plays can be supported by a subgame-perfect equilibrium. 
In such cases, a subset of agents opt in FL, while others opt out in every stage. By Theorem \ref{thm: social welfare maximization}, this outcome also leads to the best possible social welfare over time. %\ly{Checked}
% \ew{I rewrote the paragraph a bit, double check} 

However, achieving it in the real world would require we address the following impractical assumptions: 1) Players need to have full knowledge about the system, including other players' decisions, which is impractical in privacy-preserving FL systems; 2) A unique choice of equilibrium is %achieved
required through some central coordination. Due to the existence of multiple type 2 equilibria in the stage game, simultaneous decision-making without coordination might end up with a non-equilibrium play in the first stage; 3) Players need to have unlimited cognitive ability and computation power, i.e., fully rational, to be able to capture the subgame perfect equilibrium, which rarely holds in real-world applications, especially in cross-device FL, where human or device limitations can lead to suboptimal decision-making.

\subsection{Myopic Strategy}

To address the drawbacks of the subgame-perfect equilibrium strategy, we propose a myopic strategy that is privacy-preserving, computationally efficient, and memory-light. Notably, this strategy converges to a {neighborhood of} Nash equilibrium of the stage game $\mathcal{G}$ in finite time. Since we have shown in Section \ref{subsec: social welfare} that an equilibrium aligns with an optimal social welfare strategy, the myopic strategy also provides a practical mechanism that drives the system toward near-optimal social welfare outcomes in decentralized FL settings. In this subsection, we first present the strategy in Algorithm 1, then establish the convergence properties, and finally give a few remarks regarding the practical use of the strategy.

\IncMargin{1em}
\begin{algorithm}
\caption{The myopic strategy}\label{alg: myopic strategy}
 % \SetAlgoLined
\Init{agents randomly choose $s^0$}{}
\For{$t=1,2,\dots$}{
server broadcasts $|\Omega^{t-1}|$, $\bar{\mu}^{t-1}=\sum_{j \in \Omega^{t-1}}\mu_j/|\Omega^{t-1}|$ (Set $\bar{\mu}^{t-1} = 0$ if $|\Omega^{t-1}| = 0$)\\
each agent $i$ calculates:\\
\If{$s_i^{t-1}=0$}{
    $c_i^t(\bar{s}_i^{t-1}=1,s_{-i}^{t-1}) \gets \frac{a}{(|\Omega^{t-1}|+1)n} + \big| \mu_i-\frac{|\Omega^{t-1}|\bar{\mu}^{t-1}+\mu_i}{|\Omega^{t-1}|+1} \big|$}
\ElseIf{{$s_i^{t-1}=1$}}
    {$c_i^t(\bar{s}_i^{t-1}=0,s_{-i}^{t-1}) \gets \frac{a}{n}$}
    % \ew{what about $s_i^{t-1}=1$?}
each agent $i$ chooses:\\
\If{$c_i^t(\bar{s}_i^{t-1},s_{-i}^{t-1}) < c_i^{t-1}(s_i^{t-1},s_{-i}^{t-1})$\label{line:inertia}}{
    $s_i^t = \bar{s}_i^{t-1}$}
\Else{}
{
    $s_i^t = s_i^{t-1}$}
}
\end{algorithm}

%Essentially, t
The myopic strategy uses the principle of ``best-reply dynamics'' \cite{best_reply}. After observing the costs realized in the previous stage, all agents use the server's broadcast information to evaluate the hypothetical costs they would have incurred by unilaterally choosing a different action. If their hypothetical costs are lower than the realized costs, agents will change their decisions simultaneously at the current stage.\footnote{Note that it is different from the conventional sequential best-response dynamics \cite{best_response}, where only one player updates their strategy at each stage. We will discuss more details later in the subsection.} We first summarize the updating principles of agents' decisions. 

\begin{lemma}\label{lemma: deviation condition}
    Suppose all players conduct the myopic strategy. For any stage $t > 0$, given the broadcast information $|\Omega^{t-1}| > 0$ and $\bar{\mu}^{t-1}$, the actions of players at stage $t$ are determined as follows,
    \begin{enumerate}
        \item for any non-FL agent $i$ at stage $t-1$: chooses $s_i^{t} = 1$ if $|\mu_i - \bar{\mu}^{t-1}| < \frac{a}{n}$ and $s_i^{t} = 0$ otherwise;
        \item for any FL agent $i$ at stage $t-1$: chooses $s_i^{t} = 0$ if $|\mu_i - \bar{\mu}^{t-1}| > \frac{a}{n}(1-\frac{1}{|\Omega^{t-1}|})$ and $s_i^{t} = 1$ otherwise.
    \end{enumerate}
\end{lemma}
\begin{proof}
    For any non-FL agent $i$ at stage $t-1$, i.e., $s_i^{t-1} = 0$, they will opt in FL if and only if $c_i^t(\bar{s}_i^{t-1} = 1,s_{-i}^{t-1}) < c_i^{t-1}(s_i^{t-1} = 0,s_{-i}^{t-1})$. According to the myopic strategy, this is equivalent to 
    \begin{align*}
        \frac{a}{(|\Omega^{t-1}|+1)n} + \Big| \mu_i-\frac{|\Omega^{t-1}|\bar{\mu}^{t-1}+\mu_i}{|\Omega^{t-1}|+1} \Big| < \frac{a}{n}.
    \end{align*}
    Rearrange the terms, we obtain 
    $$\frac{|\Omega^{t-1}| |\mu_i - \bar{\mu}^{t-1}|}{|\Omega^{t-1}| + 1} < \frac{|\Omega^{t-1}|a}{(|\Omega^{t-1}|+1)n},$$
    which can be further simplified as $|\mu_i - \bar{\mu}^{t-1}| < \frac{a}{n}$ since $|\Omega^{t-1}| > 0$.

    For any FL agent $i$ at stage $t-1$, i.e., $s_i^{t-1} = 1$, they will opt out of FL if and only if $c_i^t(\bar{s}_i^{t-1} = 0,s_{-i}^{t-1}) < c_i^{t-1}(s_i^{t-1} = 1,s_{-i}^{t-1})$. In other words, 
    $$\frac{a}{n} < \frac{a}{|\Omega^{t-1}|n} + |\mu_i - \bar{\mu}^{t-1}|.$$
    Equivalently, we have $|\mu_i - \bar{\mu}^{t-1}| > \frac{a}{n}(1-\frac{1}{|\Omega^{t-1}|})$.
\end{proof}

Based on Lemma \ref{lemma: deviation condition} and assuming $|\Omega^{t-1}| > 0$, we have the following observations about the myopic strategy: \textbf{a) Non-FL agent behavior}: non-FL agents make decisions independently of $|\Omega^{t-1}|$. While it may seem intuitive that a larger number of existing FL participants would encourage non-FL agents to join, this is true when agents have identical distribution means. {However, for any data separation $\Delta>0$, there is a trade-off between the total number of data samples in the FL system (first term in Eq. \eqref{eq: cost join}) and the cost induced by different data qualities (second term in Eq. \eqref{eq: cost join}).}
% \ew{Let's not call it model bias, that has a specific meaning the learning context. Also I thought the tradeoff is between mean distance and the local data sample size. } 
{Though fewer FL participants at stage $t-1$ indicate fewer total data samples, it allows a non-FL agent to easily manipulate the global model with their own data if they want to opt in FL. Overall, non-FL agents are indifferent to how many FL participants are already in the FL system;} 
% \ew{I don't understand the last sentence} 
\textbf{b) FL agent behavior}: FL agents consider both $|\Omega^{t-1}|$ and $\bar{\mu}^{t-1}$ when making decisions. {A larger $|\Omega^{t-1}|$ indicates more data contributed to FL, thus the cost related to the total number of data samples is reduced.} As a result, a larger number of FL participants in the previous stage can discourage FL agents from leaving; \textbf{c) Inertia in player decisions}: {observed from line \ref{line:inertia} of Algorithm \ref{alg: myopic strategy},} players exhibit inertia, meaning that they always keep their previous actions when they are indifferent at the current stage. This feature helps stabilize the system outcomes. 

Since non-FL and FL agents make decisions based on different criteria, tracking the strategy profile across stages is challenging. However, if some particular states occur, the learning dynamics later on can be analyzed more easily.

\begin{lemma}\label{lemma: myopic dynamic}
    Suppose all players conduct the myopic strategy and recall $k^*$ is the number of FL participants in a type 2 equilibrium of the stage game $\mathcal{G}$. Assume $s^{t-1}$ forms a $|\Omega^{t-1}|$-consecutive participation for some $t > 0$ and $|\Omega^{t-1}|>0$, then $s^t$ also forms a $|\Omega^t|$-consecutive participation. Moreover, at stage $t$,
    \begin{enumerate}
        \item if $|\Omega^{t-1}| = k^*$: $s^t = s^{t-1}$;
        \item if $|\Omega^{t-1}| < k^*$: some non-FL agent(s) joins, no FL agent leaves, and $|\Omega^t| \ge |\Omega^{t-1}| + 1$;
        \item if $|\Omega^{t-1}| > k^*$: no non-FL agent joins, some FL agent(s) leaves, and $|\Omega^t| = k^*$ or $k^* \pm 1$.
    \end{enumerate}
\end{lemma}
\begin{proof}
    If $|\Omega^{t-1}| = k^*$ and $s^{t-1}$ forms a $k^*$-consecutive participation, we observe that $s^{t-1}$ is then an equilibrium strategy of the stage game $\mathcal{G}$. By the constructive idea of the myopic strategy, no agent benefits from unilaterally deviating, and thus no agent changes their strategy. Therefore, $s^{t} = s^{t-1}$.

    If $|\Omega^{t-1}| < k^*$, for any $i \in \Omega^{t-1}$, we have 
    $$|\mu_i - \bar{\mu}^{t-1}| \le \frac{|\Omega^{t-1}|-1}{2}\Delta \le \frac{|\Omega^{t-1}|-1}{2} \frac{2a}{k^*n} = \frac{a}{n}\frac{|\Omega^{t-1}|-1}{k^*} \le \frac{a}{n}\big(1-\frac{1}{|\Omega^{t-1}|}\big).$$
    The first inequality is because $\Omega^{t-1}$ is a $|\Omega^{t-1}|$-consecutive set, the second inequality comes from Theorem \ref{thm: NE} and Assumption \ref{assumption: odd k}, and the last inequality is due to $1 \le |\Omega^{t-1}| < k^*$. Thus, by Lemma \ref{lemma: deviation condition}, we see that no FL agent will leave the FL system if $|\Omega^{t-1}| < k^*$. Similarly, for $i \notin \Omega^{t-1}$, we have
    $$\min_{i \notin \Omega^{t-1}} |\mu_i - \bar{\mu}^{t-1}| = \frac{|\Omega^{t-1}|+1}{2}\Delta < \frac{|\Omega^{t-1}|+1}{2} \frac{2a}{k^*n} \le \frac{a}{n}.$$
    Therefore, by Lemma \ref{lemma: deviation condition}, there exists at least one non-FL agent who will opt in FL at stage $t$, and then we conclude $|\Omega^t| \ge |\Omega^{t-1}| + 1$. We also observe that, if a non-FL agent $j$ opts in FL, then so should all the other non-FL agents with indices between $j$ and any index in $\Omega^{t-1}$. Thus, $\Omega^t$ is also a $|\Omega^t|$-consecutive set.  

    If $|\Omega^{t-1}| > k^*$, we first show that no non-FL agent will opt in FL at stage $t$. This is trivial when $|\Omega^{t-1}| = m$. When $|\Omega^{t-1}| \neq m$, there exists non-FL agent at stage $t-1$. We then have
    $$\min_{i \notin \Omega^{t-1}} |\mu_i - \bar{\mu}^{t-1}| = \frac{|\Omega^{t-1}|+1}{2}\Delta > \frac{|\Omega^{t-1}|+1}{2} \frac{2a}{(k^*+1)n} > \frac{a}{n},$$
    where the first inequality is based on Theorem \ref{thm: NE} and the second inequality is due to $|\Omega^{t-1}| > k^*$. Hence, no non-FL agent will opt in FL at stage $t$. Since $s^{t-1}$ is not an equilibrium, at least one player benefits from unilaterally deviating. Based on the idea of the myopic strategy and the fact that no non-FL agent deviates, there has to be some FL agent(s) opting out of FL at stage $t$. We next show that $|\Omega^t| = k^*$ or $k^* \pm 1$. Based on Theorem \ref{thm: NE}, we have
    \begin{align}\label{eq: temp2}
        \begin{dcases}
            \frac{2a}{n}\big(1-\frac{1}{|\Omega^{t-1}|}\big) > \frac{2a}{n}\big(1-\frac{1}{k^*}\big) \ge k^*\Delta\big(1-\frac{1}{k^*}\big) = (k^*-1)\Delta,\\
            \frac{2a}{n}\big(1-\frac{1}{|\Omega^{t-1}|}\big) < \frac{2a}{n} < (k^*+1)\Delta.
        \end{dcases}
    \end{align}
    According to Lemma \ref{lemma: deviation condition}, a FL agent $i$ deviates if and only if $|\mu_i - \bar{\mu}^{t-1}| > \frac{a}{n}(1-\frac{1}{|\Omega^{t-1}|})$, which means a FL agent stays as a FL participant as long as 
    $$|\mu_i - \bar{\mu}^{t-1}| \le \frac{a}{n}(1-\frac{1}{|\Omega^{t-1}|}).$$ From Eqs. \eqref{eq: temp2}, we observe that at least $k^*-1$ and at most $k^*+1$ agents satisfy the preceding condition. Therefore, $|\Omega^t| = k^*$ or $k^* \pm 1$. The condition also indicates that $\Omega^t$ is a $|\Omega^t|$-consecutive set. 
\end{proof}

Assuming the strategy forms consecutive participation at some stage, Lemma \ref{lemma: myopic dynamic} shows how the strategy profile evolves from one stage to the next under the myopic strategy. Building on this, the next lemma extends the analysis to the long-term behavior of the system. Specifically, it characterizes the eventual outcome after multiple stages, showing that the strategy profiles converge to a neighborhood of Nash equilibrium of the stage game $\mathcal{G}$. Together, these results provide a step-by-step and long-term perspective on the dynamics of the myopic strategy. First, we define a {\textit{neighborhood of Nash equilibrium}} as follows,
\begin{definition}
    Let $s^*$ be a type 2 equilibrium strategy of the stage game $\mathcal{G}$ with exactly $k^*$ FL participants. We define the {neighborhood of equilibrium} as the set of strategy profiles that include not only $s^*$ but also {some} that differ from $s^*$ {by at most one agent's strategy},
    \begin{align*}
        \mathcal{N}(s^*) = s^* \cup \big\{s \in S: \text{$s$ forms a $(k^* \pm 1)$-consecutive participation and }||s-s^*||_1=1 \big\}.
    \end{align*}
\end{definition} 
That is, {a neighborhood of equilibrium} consists of strategy profiles where the number of consecutive FL participants is either $k^*$, $k^*-1$ or $k^*+1$, and $s$ differs from $s^*$ in exactly one agent's participation choice.

\begin{lemma}\label{lemma: convergence}
    Suppose all players conduct the myopic strategy. Assume $s^{t-1}$ form a $|\Omega^{t-1}|$-consecutive participation for some $|\Omega^{t-1}| > 0$. The strategy profile $s^t$ converges to {a neighborhood of equilibrium} $\mathcal{N}(s^*)$ in finite time, where $s^*$ is some type 2 equilibrium of the stage game $\mathcal{G}$.
\end{lemma}
\begin{proof}
    We consider the following cases separately, based on different values of $|\Omega^{t-1}|$. 

    \paragraph{Case (I): $|\Omega^{t-1}| = k^*$}
    Based on Lemma \ref{lemma: myopic dynamic} and since we assume $s^{t-1}$ forms a $|\Omega^{t-1}|$-consecutive participation, we obtain $s^t = s^{t-1}$. By induction, there holds $s^{t-1} = s^t = s^{t+1} = \cdots = s^{t+T}$ for any $T \ge 0$. Thus, $s^t$ converges to an exact equilibrium of the stage game $\mathcal{G}$.

    \paragraph{Case (II): $|\Omega^{t-1}| = k^*+1$}
    Denote the FL agents with the smallest and largest index by $\ell$ and $r$, respectively. It is straightforward to see $\arg\max_{i \in \Omega^{t-1}} |\mu_i - \bar{\mu}^{t-1}| = \{\ell, r\}$. By Lemmas \ref{lemma: deviation condition} and \ref{lemma: myopic dynamic}, agents $\ell$ and $r$ will opt out of FL at stage $t$ while all the other agents keep their choices unchanged. Thus, $s^t$ forms a $|\Omega^t|$-consecutive participation with $|\Omega^t| = k^*-1$ and $\bar{\mu}^t = \bar{\mu}^{t-1}$. Note that $s^t$ is the same as $s^{t-1}$ except for two agents $\ell$ and $r$. 
    
    Furthermore, at stage $t+1$, if $k^* = 1$, then we have $s^t = 0$. According to the design of the myopic strategy, we then have $s^t = 0$ for any $t > 0$ thereafter. Thus, by definition, $s^t$ converges to a neighborhood of equilibrium. If $k^* > 1$, based on Lemma \ref{lemma: myopic dynamic} and the fact that $|\mu_\ell - \bar{\mu}^t| = |\mu_r - \bar{\mu}^t|$, agents $\ell$ and $r$ will rejoin FL again. Moreover, we have 
    $$\min_{i \notin \Omega^t \cup \{\ell, r\}} |\mu_i - \bar{\mu}^t| = \big(\frac{k^*-2}{2} + 2\big)\Delta > \frac{k^*+2}{2} \frac{2a}{(k^*+1)n} > \frac{a}{n}.$$
    Therefore, by Lemma \ref{lemma: deviation condition}, no non-FL agent other than $\ell$ and $r$ will opt in FL at stage $t+1$. Hence, $s^{t+1} = s^{t-1}$. Repeat the same arguments, we will have $s^{t-1} = s^{t+1} = \cdots = s^{t+T}$ for any odd $T > 0$ and $s^{t} = s^{t+2} = \cdots = s^{t+T}$ for any even $T > 0$. {Essentially, the strategy $s^{t-1}$ says all agents with indices in between $\ell$ and $r$ choose to opt in FL including agents $\ell$ and $r$, while others opt out. In contrast, the strategy $s^{t}$ says all agents with indices strictly between $\ell$ and $r$ choose to opt in FL, while others opt out. Since $|\Omega^{t-1}| = k^*+1$ and $|\Omega^t| = k^*-1$, we have $s^{t-1}, s^t \in \mathcal{N}(s^*)$ for some type 2 equilibrium $s^*$.} 
    % \ew{describe $s^t$ and $s^{t-1}$ in words, all agents strictly between $l$ and $r$ choose to join the system,  all agents in between, including $l$ and $r$ ...} 
    Therefore, the strategy profile $s^t$ converges to a neighborhood of equilibrium. 

    \paragraph{Case (III): $|\Omega^{t-1}| = k^*-1$}
    If $k^* = 1$, we then have $s^t = 0$ for any $t > 0$ and thus $s^t$ converges to a neighborhood of equilibrium. For $k^* > 1$, pick $\ell \in \{1,2,...,m\}$ such that $s_\ell^{t-1} = 0$ but $s_{\ell+1}^{t-1} = 1$, if it exists, and pick $r \in \{1,2,...,m\}$ such that $s_r^{t-1} = 0$ but $s_{r-1}^{t-1} = 1$, if it exists. Since $|\Omega^{t-1}| = k^*-1$ and $\Omega^{t-1}$ is a $|\Omega^{t-1}|$-consecutive set, at least one of $\ell$ and $r$ exists. Following the arguments in Case (II), we can conclude that $s^t$ forms a $|\Omega^t|$-consecutive participation with $|\Omega^t| = k^*$ or $k^*+1$, which corresponds to Case (I) or Case (II) thereafter, respectively.

    \paragraph{Case (IV): $|\Omega^{t-1}| > k^*+1$}
    By Lemma \ref{lemma: myopic dynamic}, we have $|\Omega^{t}| = k^*$ or $k^* \pm 1$, and $s^{t}$ also forms a $|\Omega^t|$-consecutive participation. Thereafter, the strategy dynamics have been studied in Case (I), (II) and (III).

    \paragraph{Case (V): $|\Omega^{t-1}| < k^* - 1$}
    According to Lemma \ref{lemma: myopic dynamic} and by induction, $|\Omega^t|$ is strictly increasing until it becomes greater or equal to $k^*$. Thereafter, it is equivalent to Case (IV). 
\end{proof}

The preceding lemma has demonstrated the convergence of the myopic strategy when some special states are observed. Finally, the following theorem establishes the convergence result for any random initial state $s^0$.

\begin{theorem}
    Suppose all players conduct the myopic strategy. Given any initial strategy $s^0$ with $|\Omega^0| > 0$, the strategy profile across stages $s^t$ converges to {a neighborhood of equilibrium} in finite time. If $|\Omega^0| = 0$, then $s^t$ retains a type 1 equilibrium of the stage game $\mathcal{G}$, i.e., $s^t = 0$ for any $t > 0$.
\end{theorem}
\begin{proof}
    When $|\Omega^0| = 0$, no one benefits by unilaterally joining the FL. Thus $s^t = 0$ for any $t > 0$.

    For $|\Omega^0| > 0$, first, we establish that there exists one stage $T$ such that $s^T$ forms a $|\Omega^T|$-consecutive participation. To do this, consider any stage $t > 0$ with the given information $|\Omega^{t-1}|$ and $\bar{\mu}^{t-1}$. Assume $\Omega^{t-1}$ is not a $|\Omega^{t-1}|$-consecutive set, then $|\Omega^{t-1}| \ge 2$.

    When $2 \le |\Omega^{t-1}| \le \frac{k^*-1}{2}$, without loss of generality, we may assume $\mu_1 = \Delta$. Since $\Omega^{t-1}$ is not a $|\Omega^{t-1}|$-consecutive set, there holds $\frac{|\Omega^{t-1}|+1}{2}\Delta < \bar{\mu}^{t-1} < (m-\frac{|\Omega^{t-1}|+1}{2})\Delta$. Moreover, we have 
    \begin{align}\label{eq: temp5}
        \frac{a}{n}\big(1-\frac{1}{|\Omega^{t-1}|}\big) \ge \frac{a}{n}\big(1-\frac{1}{2}\big) = \frac{a}{2n} > \frac{k^*}{4}\Delta,
    \end{align}
    where the last inequality is based on Theorem \ref{thm: NE}. From Lemma \ref{lemma: deviation condition}, we know that any agent $i$ satisfying $|\mu_i - \bar{\mu}^{t-1}| \le \frac{a}{n}\big(1-\frac{1}{|\Omega^{t-1}|}\big)$ will participate in the FL at stage $t$. Given the bound in Eq. \eqref{eq: temp5}, it can be seen that $s_i^t = 1$ if agent $i$ has the distribution mean $\mu_i$ satisfying
    \begin{align}\label{eq: temp6}
        \mu_i \in [\bar{\mu}^{t-1}-\frac{k^*\Delta}{4}, \bar{\mu}^{t-1}+\frac{k^*\Delta}{4}].
    \end{align}
    Depending on the value of $\bar{\mu}^{t-1}$, the number of agents satisfying the preceding condition may vary. Thus, we discuss the following cases separately. 
    
    \paragraph{Case (I): $\mu_1 \le \bar{\mu}^{t-1}-\frac{k^*\Delta}{4}$ and $\mu_m \ge \bar{\mu}^{t-1}+\frac{k^*\Delta}{4}$}
    In this case, since the data separation $\Delta = \mu_m - \mu_{m-1} = \cdots = \mu_2 - \mu_1$, the number of agents satisfying Eq. \eqref{eq: temp6} is at least 
    \begin{align*}
        \floor*{\frac{(\bar{\mu}^{t-1}+\frac{k^*\Delta}{4}) - (\bar{\mu}^{t-1}-\frac{k^*\Delta}{4})}{\Delta}} = \floor*{\frac{k^*}{2}}.
    \end{align*}

    \paragraph{Case (II): $\mu_1 > \bar{\mu}^{t-1}-\frac{k^*\Delta}{4}$ and $\mu_m \ge \bar{\mu}^{t-1}+\frac{k^*\Delta}{4}$}
    In this case, all agents with indices between 1 to $\floor{\frac{\bar{\mu}^{t-1}+\frac{k^*\Delta}{4}}{\Delta}}$ satisfy Eq. \eqref{eq: temp6}. Thus, the number of agents is precisely $\floor{\frac{\bar{\mu}^{t-1}+\frac{k^*\Delta}{4}}{\Delta}}$.

    \paragraph{Case (III): $\mu_1 \le \bar{\mu}^{t-1}-\frac{k^*\Delta}{4}$ and $\mu_m < \bar{\mu}^{t-1}+\frac{k^*\Delta}{4}$}
    This case is similar to case (II), thus the number of agents satisfying Eq. \eqref{eq: temp6} can be obtained as $m - \floor{\frac{\bar{\mu}^{t-1}-\frac{k^*\Delta}{4}}{\Delta}} + 1$.

    \paragraph{Case (IV): $\mu_1 > \bar{\mu}^{t-1}-\frac{k^*\Delta}{4}$ and $\mu_m < \bar{\mu}^{t-1}+\frac{k^*\Delta}{4}$}
    In this case, all agents satisfy Eq. \eqref{eq: temp6}, and thus they all opt in FL at stage $t$, making $s^t$ form a consecutive participation. 
    
    Here, we consider the non-trivial case where $\Omega^t$ is not a $|\Omega^t|$-consecutive set. Thus, there needs to be at least one FL agent $i$ at stage $t$ who does not satisfy $|\mu_i - \bar{\mu}^{t-1}| \le \frac{a}{n}\big(1-\frac{1}{|\Omega^{t-1}|}\big)$. Therefore, the number of FL participants at stage $t$ is at least
    $$|\Omega^t| \ge \min\Big\{ \floor*{\frac{k^*}{2}}, \floor*{\frac{\bar{\mu}^{t-1}+\frac{k^*\Delta}{4}}{\Delta}}, m - \floor*{\frac{\bar{\mu}^{t-1}-\frac{k^*\Delta}{4}}{\Delta}} + 1 \Big\} + 1.$$
   
    By Assumption \ref{assumption: odd k} and $2 \le |\Omega^{t-1}| \le \frac{k^*-1}{2}$, we have $\floor{\frac{k^*}{2}} = \frac{k^*-1}{2}$. Since $\frac{|\Omega^{t-1}|+1}{2}\Delta < \bar{\mu}^{t-1} < (m-\frac{|\Omega^{t-1}|+1}{2})\Delta$, we can bound the second term as
    $$\floor*{\frac{\bar{\mu}^{t-1}+\frac{k^*}{4}\Delta}{\Delta}} \ge \floor*{\frac{|\Omega^{t-1}|+1}{2} + \frac{2|\Omega^{t-1}|+1}{4}} = |\Omega^{t-1}|.$$
    Similarly, the third term can be lower bounded as well,
    $$m - \floor*{\frac{\bar{\mu}^{t-1}-\frac{k^*\Delta}{4}}{\Delta}} + 1 \ge \floor*{\frac{|\Omega^{t-1}|+1}{2} + \frac{2|\Omega^{t-1}|+1}{4}} = |\Omega^{t-1}|.$$
    Then, we obtain
    $$|\Omega^t| \ge \min\Big\{ \frac{k^*-1}{2}, |\Omega^{t-1}| \Big\} + 1 \ge |\Omega^{t-1}|+1.$$
    Thus, the number of FL participants is strictly increasing until $|\Omega^T| \ge \frac{k^*+1}{2}$ for some $T$, or $|\Omega^T|$ becomes a $|\Omega^T|$-consecutive set. 

    When $|\Omega^{t-1}| \ge \frac{k^*+1}{2}$, there holds
    $$\frac{a}{n}\big(1-\frac{1}{|\Omega^{t-1}|}\big) \ge \frac{a}{n}\big(1-\frac{2}{k^*+1}\big) = \frac{a}{n} - \frac{2a}{(k^*+1)n} > \frac{a}{n} - \Delta,$$
    where the last inequality is because of Theorem \ref{thm: NE}. Based on Lemma \ref{lemma: deviation condition}, we observe that $\Omega^t$ must be a $|\Omega^t|$-consecutive set. 

    Therefore, there exists a stage $T$ such that $s^T$ forms a $|\Omega^T|$-consecutive participation and $T \le \frac{k^*+1}{2}$. Once a consecutive participation is formed, we apply Lemma \ref{lemma: convergence} and the proof is then complete. 
\end{proof}

\begin{figure}[htbp]
    \centering
    \begin{subfigure}{\textwidth}
        \begin{subfigure}[b]{0.48\textwidth}
            \includegraphics[width=\textwidth]{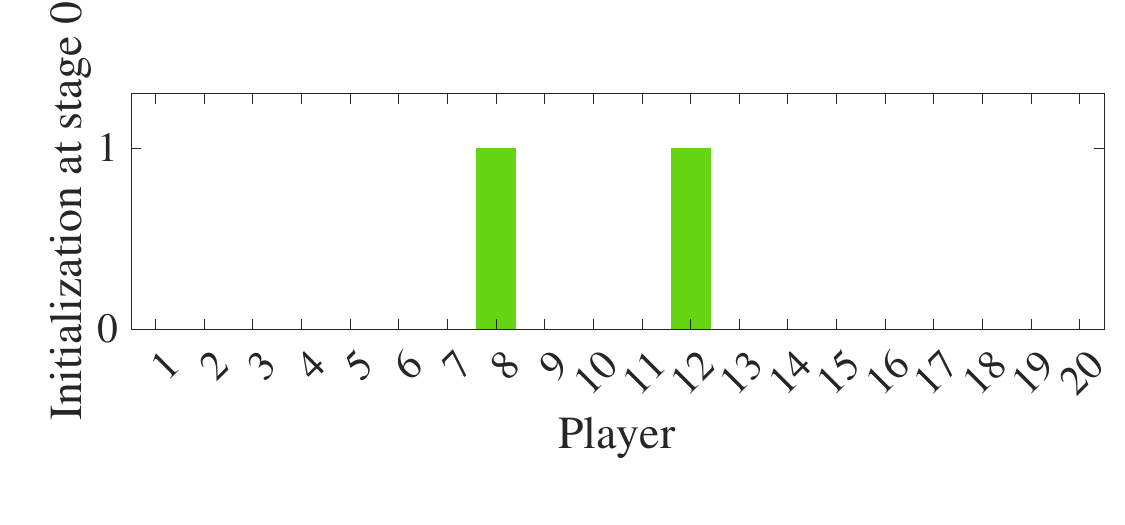}
        \end{subfigure}
        \hfill
        \begin{subfigure}[b]{0.48\textwidth}
            \includegraphics[width=\textwidth]{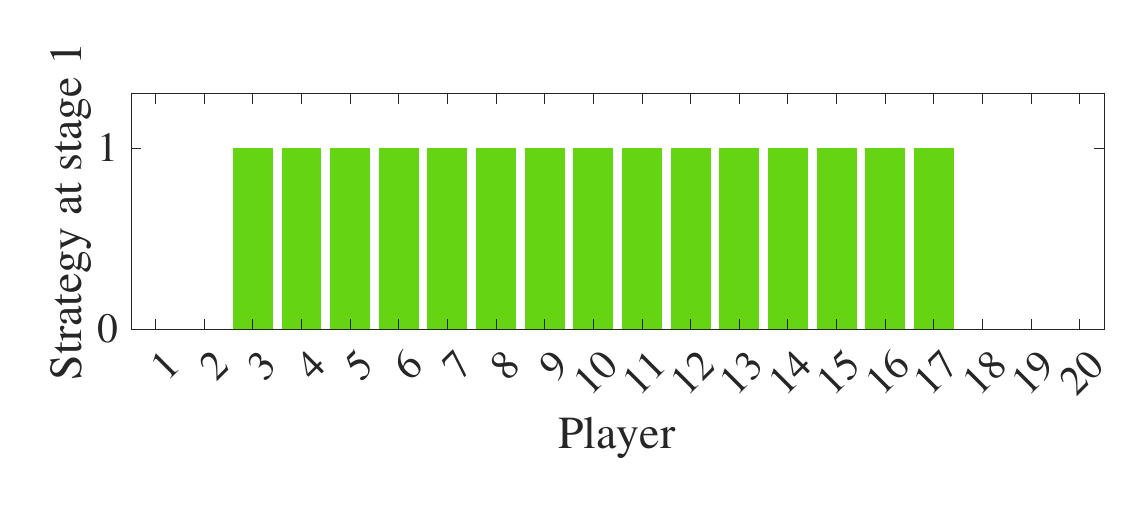}
        \end{subfigure}
        \vspace{-0.5cm}
        \subcaption{sparse initialization}
        \label{fig: sparse}
    \end{subfigure}
    
    \begin{subfigure}{\textwidth}
        \begin{subfigure}[b]{0.48\textwidth}
            \includegraphics[width=\textwidth]{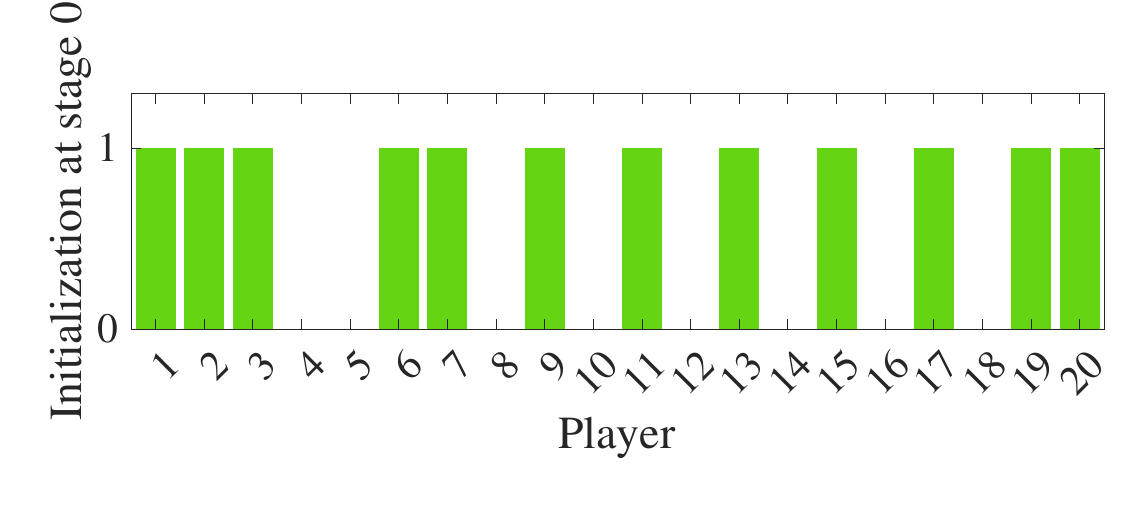}
        \end{subfigure}
        \hfill
        \begin{subfigure}[b]{0.48\textwidth}
            \includegraphics[width=\textwidth]{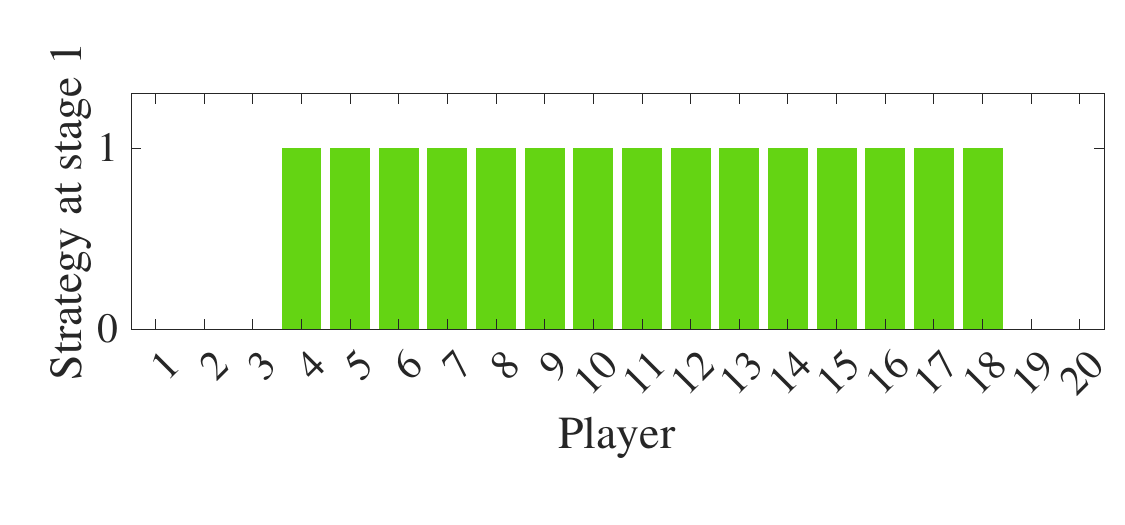}
        \end{subfigure}
        \vspace{-0.5cm}
        \subcaption{dense initialization}
        \label{fig: dense}
    \end{subfigure}
    \caption{Examples of the myopic strategy dynamics with two different initializations. We use parameters $m=20$, $n=100$, $a=790$. The number of FL participants of a type 2 equilibrium is $k^*=15$.}
    \label{fig: myopic}
\end{figure}

Fig. \ref{fig: myopic} illustrates the dynamics of the myopic strategy under two different initializations. The $x$-axis represents the players, indexed from 1 to 20. The $y$-axis shows whether or not each player opts in FL, where a bar of height 1 indicates participation in FL, and 0 indicates non-participation. Fig. \ref{fig: sparse} starts with a sparse initialization at stage 0 with only a few players opting in FL, while Fig. \ref{fig: dense} begins with a more evenly distributed initialization with many players opting in. Both examples demonstrate the efficiency of the myopic strategy in driving the system toward an equilibrium state. The fast convergence highlights its practical applicability in real-world scenarios.

\paragraph{Privacy-preserving design} One of the key advantages of the proposed myopic strategy is its inherent privacy-preserving design, which aligns seamlessly with the principles of federated learning. The strategy only relies on limited and aggregated information broadcast by the central server---specifically, the number of FL participants in the previous stage and the mean distribution of all participants' data. Unlike approaches such as the subgame-perfect equilibrium \cite{fudenberg_1991_game}, this strategy ensures that individual agents are not required to disclose sensitive information, such as private data distributions or realized costs. Therefore, data privacy can be maintained, and meanwhile, the strategy drives the system close to an equilibrium in finite time. 

\paragraph{Decentralized decision-making} Another notable aspect of the myopic strategy is its decentralized decision-making process, which minimizes the need for coordination. In this approach, agents update their strategies simultaneously at each stage. This is particularly important in federated learning, where a central coordinator may not be feasible. In contrast, sequential best-response dynamics \cite{best_response} enforce one and only one agent to change his/her strategy at each stage and others must wait. Despite its well-studied convergence properties \cite{BR_convergence}, this setting has two main drawbacks. First, a coordinator is required to schedule the moves of every agent. Second, the convergence process may be significantly slowed down, and inefficiencies may occur. 

\paragraph{Efficiency and practical convergence}
The myopic strategy is computationally and memory efficient. By using simple updating rules based on local cost observations and limited system information, the strategy converges to a neighborhood of equilibrium in finite time. This makes it feasible for real-world applications, particularly in resource-constrained federated learning environments \cite{Imteaj}. Moreover, empirical observations (e.g., Fig. \ref{fig: myopic}) show that the strategy exhibits fast convergence in practice, this further validates its efficiency. 

\paragraph{Bounded Rationality} The myopic strategy accounts for bounded rationality \cite{bounded_rationality}, since it requires only simple decision-making from agents rather than complex reasoning processes. A fully rational agent, for instance, might prefer strategies such as Bayesian updates, which involve calculating and updating probabilistic beliefs about the actions of other players based on all past observations \cite{Brandt, wu2021multiagentbayesianlearningbest}. Though such an approach can theoretically lead to better outcomes, it imposes significant computation power and cognitive abilities \cite{JasonMarden}, which makes it impractical in some large-scale, decentralized systems like federated learning with humans involved. In contrast, the myopic strategy replies on straightforward cost comparisons and uses limited aggregated information. Moreover, beyond proposing a simple yet efficient strategy, we also seek to capture how actual human agents might behave in a multi-stage federated learning scenario. By embracing bounded rationality, the myopic strategy provides a more realistic and practical approach that characterizes human decision-making in real-world applications.

\section{Conclusion}
This work investigates the fundamental problem of FL participation, with a focus on the impact of heterogeneous data qualities. By modeling FL participation as a stage game, we find that Nash equilibrium exists and coincides with the social welfare maximum strategy. Extending the framework into a repeated game, we analyze long-term FL participation dynamics. We propose a myopic strategy that efficiently guides the system close to an equilibrium of the stage game in finite time, and meanwhile preserves privacy and requires minimal computational resources. It highlight the feasibility of achieving stable and socially optimal outcomes in decentralized FL environments. In future works, we aim to consider a generalized payoff/cost structure non-homogeneous data separation to accommodate real-world settings. 

% Bibliography
\bibliographystyle{ACM-Reference-Format-num}
\bibliography{sample-bibliography}

\end{document}